\documentclass[conference]{IEEEtran}
\makeatletter
\def\ps@headings{%
\def\@oddhead{\mbox{}\scriptsize\rightmark \hfil \thepage}%
\def\@evenhead{\scriptsize\thepage \hfil \leftmark\mbox{}}%
\def\@oddfoot{}%
\def\@evenfoot{}}
\makeatother
\newtheorem{theorem}{Theorem}[section]
\newtheorem{lemma}[theorem]{Lemma}
\newtheorem{proposition}[theorem]{Proposition}
\newtheorem{corollary}[theorem]{Corollary}

\usepackage{bm}
\usepackage[noend]{algorithmic}
\usepackage{algorithm}
\usepackage{color}
\usepackage{epsfig}
\usepackage{amssymb}
\usepackage{amsmath}
\usepackage{amsfonts}
\usepackage{setspace}

\usepackage{multirow}

\usepackage[tight,footnotesize]{subfigure}

\begin{document}

\title{Maximizing System Throughput by Cooperative Sensing in Cognitive Radio Networks}

\author{\IEEEauthorblockN{Shuang Li}
\IEEEauthorblockA{
Email:
li.908@osu.edu}
\and
\IEEEauthorblockN{Zizhan Zheng}
\IEEEauthorblockA{
Email:
zhengz@ece.osu.edu}
\and
\IEEEauthorblockN{Eylem Ekici}
\IEEEauthorblockA{
Email:
ekici@ece.osu.edu}
\and
\IEEEauthorblockN{Ness Shroff}
\IEEEauthorblockA{
Email:
shroff@ece.osu.edu}}
\maketitle

\begin{abstract}
Cognitive Radio Networks allow unlicensed users to
opportunistically access the licensed spectrum without causing
disruptive interference to the primary users (PUs). One of the main
challenges in CRNs is the ability to detect PU transmissions. Recent
works have suggested the use of secondary user (SU) cooperation over
individual sensing to improve sensing accuracy. In this paper, we
consider a CRN consisting of a single PU and multiple SUs to study the
problem of maximizing the total expected system throughput. We propose
a Bayesian decision rule based algorithm to solve the problem
optimally with a constant time complexity. To prioritize PU
transmissions, we re-formulate the throughput maximization problem by
adding a constraint on the PU throughput. The constrained
optimization problem is shown to be NP-hard and solved via a greedy
algorithm with pseudo-polynomial time complexity that achieves strictly greater than $1/2$ of the optimal solution. We also investigate
the case for which a constraint is put on the sensing time overhead,
which limits the number of SUs that can participate in cooperative
sensing. We reveal that the system throughput is monotonic over the
number of SUs chosen for sensing. We illustrate the efficacy of the
performance of our algorithms via a numerical investigation.
\end{abstract}


\section{Introduction}
\label{sec:intro}
Cognitive radio networks (CRNs) have been proposed to address the
spectrum scarcity problem by allowing unlicensed users (secondary users, SUs) to access licensed spectrum on the condition of not disrupting the
communication of licensed users (primary users, PUs).
To this end, SUs sense licensed channels to detect the primary user
(PU) activities and find the underutilized ``white spaces''. 
FCC has opened the TV bands for unlicensed access \cite{FCC}\footnote{The recent
FCC ruling requires the use of central TV Band usage databases to verify 
spectral availability. While respecting this ruling, our work explores local
cooperative methods to improve sensing accuracy with the potential outcome of
relieving this burdensome requirement.}, and IEEE has formed a
working group (IEEE 802.22 \cite{ieee}) to regulate the unlicensed access without interference. Many other organizations are also making efforts on the spectrum access policy in the CRN environment, e.g., DARPA's `Next Generation' (XG) program \cite{4221472} mandates cognitive radios to sense signals and prevent interference to existing military and civilian radio systems. To avoid the interference to PUs, sensing becomes
an indispensable part of CRN design.

Sensing can be performed via several methods, including
energy detection, cyclostationary feature detection, and compressed sensing \cite{survey}. Energy detection is a simple method 
and requires no a priori knowledge of PU signals \cite{1447503}. Its main disadvantage is its decreased accuracy in face of fading, 
shadowing, and unknown noise power profiles. For instance, if an SU 
suffers from shadowing or heavy fading, the sensed signal tends to be 
weak while the PU is transmitting, leading to incorrect decisions. 
To address these problems while maintaining sensing simplicity, 
cooperative sensing schemes that fuse the sensing results of multiple SUs 
have been proposed \cite{Ganesan05}\cite{Mishra06}\cite{5169958}.  

Cooperative sensing overcomes shortcomings of individual sensing results 
by jointly processing observations. SUs in a locality report their 
individual sensing results, which are then used in a predefined
decision rule to optimize an objective function.
Examples of such functions include maximizing sensing accuracy (generally,
a function of false alarm probability and mis-detection probability) or
maximizing the system throughput. Aside from maximizing sensing accuracy related metrics, cooperative sensing schemes are also designed to estimate the maximum transmit power for SUs so that they do not cause disruptive interference to PUs \cite{4927466}. On the other hand, cooperative sensing 
incurs additional sensing delay viz a viz individual sensing.

Three main categories of decision rules have been identified in \cite{survey}: \emph{Soft Combining, Quantized Soft Combining,} and \emph{Hard
Combining}. In the first two categories, the sensing results are sent
to the fusion center with little or no processing, while in the last
one, binary local decisions are usually reported.  Similar to sensor
networks, linear fusion rules are widely applied to achieve a
cooperative decision, such as AND, OR and majority rules \cite{5169958}. In
addition, a more advanced fusion technique that utilizes statistical
knowledge \cite{4453896} has been devised to capture the correlation between SUs
in cooperative sensing. However, the resulting algorithm is suboptimal
and its approximation factor is unknown. None of the above-mentioned
works identify optimal decision rules for general decision structures
and they require decision rules to assume particular forms (e.g., linear)
for optimality analysis.

In this paper, we design an optimal data fusion rule to (hard)
combining of the reported sensing result. More specifically, we aim to
maximize the system throughput in a CRN composed of a single PU (i.e.,
single channel) and several SUs. While the target system is a
simplified one, it is helpful in revealing the challenges associated
with the design of optimal fusion rules.  Moreover, the resulting
algorithms can easily be generalized to more complex systems comprised
of multiple channels, where sensing decisions are made per
channel. Our main contributions can be summarized as follows:

\begin{itemize}
\item In contrast to previous works that restrict the class of fusion rules, we propose a Bayesian decision rule based algorithm
to solve the throughput maximization problem optimally
with constant time complexity.
 
\item To guarantee resources for the PU, we re-formulate the
problem by adding a constraint on the PU throughput.
This constrained problem is shown to be NP-hard by reducing the
classical partition problem \cite{Garey} to it. A greedy algorithm is
obtained with pseudo-polynomial time complexity. This
approximation algorithm is analytically shown to achieve strictly greater than $1/2$ of the optimal solution.

\item We investigate systems where limited sensing overhead
is allowed, i.e., the number of sensing SUs is
restricted. Our theoretical results show
that the performance of cooperative sensing is monotonic
over the number of SUs used for sensing. However, the characterization
of the upper bound on the required number of sensing nodes remains
elusive.
  
\end{itemize}    

The paper is organized as follows: Related work is presented in Section~\ref{sec:related}. In Section~\ref{sec:model}, the system model is introduced. The system throughput maximization problem is formulated in Section~\ref{sec:max}, and solved optimally via Bayesian decision rule. In Section~\ref{sec:PU_thru}, the constrained maximization problem is formulated, which is shown to be NP-hard. A pseudo-polynomial time greedy algorithm is proposed with an approximation factor strictly greater than $1/2$. Another direction is considered in Section~\ref{sec:uncertain_set} where the system throughput is maximized subject to a constraint on the number
of sensing SUs used. 
In Section~\ref{sec:simu}, numerical results are presented for the performance of our algorithms. The paper is concluded in Section~\ref{sec:con}.  

\section{Related Work}
\label{sec:related}
Cooperative sensing solutions have been investigated in recent years. They rely on multiple SUs to exchange sensing results or a central controller to collect the sensing results from SUs. The network is usually divided into clusters and each cluster head makes the decision on the channel occupancy. Collaborations among SUs have been shown to improve the efficiency of spectrum access and allow the relaxation of constraints at individual SUs \cite{5054703}\cite{1542650}. One branch of the papers in cooperative sensing assume that the length of sensing time at individual SUs is proportional to the sensing accuracy. However, longer sensing time decreases the transmission time. The trade-off is called the \emph{sensing efficiency} problem and is discussed in \cite{Lee_2008} and \cite{Min}. In our work, we assume the observation time at each SU is fixed so that the individual sensing accuracy does not depend on it. We focus on the optimal decision rule based on the sensing results collected. 

Decision rules so far mainly focus on AND, OR, majority rules and other linear rules. Zhang et. al. \cite{4533677} show that the optimal fusion rule to minimize the cooperative sensing error rate is the half-voting rule in most cases. They show that AND or OR 
rules are optimal only in rare cases. However, other rules with more 
complicated forms have not been considered in \cite{4533677}. Based on these observations, a fast spectrum sensing algorithm is proposed for a large network where not all SUs are required for sensing while satisfying a given error bound. However, the optimal number of sensing nodes and the complexity of this problem have not been discussed. In \cite{5169958}, the SU throughput is maximized subject to sufficient protection provided to PUs. The optimal $k$-out-of-$N$ fusion rule is determined and the sensing/throughput trade-off is also analyzed. As in \cite{4533677}, no fusion rules of general forms are considered. Thus, the optimization is restricted to a small fusion rule domain. Shahid et. al. \cite{Shahid_Kamruzzaman_2010} consider the spatial variation of SUs and the fusion rule is a weighted combination of SU observations. The weight depends on the received power and path loss at each SU. Though more advanced than AND, OR, and majority rules, the weighted form is restricted to the linear function domain. In \cite{Fan_Jiang_2010}, optimal multi-channel cooperative sensing algorithms are considered to maximize the SU throughput subject to per channel detection probability constraints. The resulting non-convex problem is solved by an iterative algorithm. Compared to \cite{Fan_Jiang_2010}, our work focuses on the maximization of the system throughput, including the PUs and SUs. Although we only consider a single-channel network, which is a simplification made on the model, our decision algorithms can be applied for each channel individually. Moreover, a soft decision rule is considered in \cite{Fan_Jiang_2010}, which requires significant amount of data to be transmitted to the coordinator while our hard decision rule requires only one bit sent from each SU. 

\section{System Model}
\label{sec:model}
We consider a time-slotted cognitive radio network in which a PU network, consisting of a PU base station (PU-BS) and PU receivers, co-exists in the same area with an SU base station (SU-BS) and $M$ SUs (Figure~\ref{fig:model}). We focus on the PU transmissions over a particular channel. We consider uplink part for the SU system, i.e., only one SU can be active and transmit to the SU-BS at any given time. Some PU receivers may lie in the interference range of SUs such as PU $1$ in Figure~\ref{fig:model}. Any transmission from these SUs such as SUs $1$, $2$, and $3$ in Figure~\ref{fig:model} may cause interference to those PU receivers. We denote the set of SUs whose uplink transmission causes interference to PU receivers by $S$ and $|S|=N$ ($M\ge N$). They are indexed from $1$ to $N$. SUs outside $S$ can use the channel to transmit at any time slot without causing interference to the PUs.     

\begin{figure}[tb]
    \begin{center}
    \setlength{\unitlength}{1in}
    \includegraphics[scale=0.4]{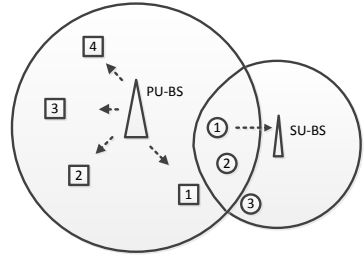}
    \end{center}
\vspace{-1.2em}
\caption{System model of an SU network overlayed with a PU network.}
\vspace{-1.8em}
\label{fig:model}
\end{figure}

SUs in $S$ are close to the PU network and they may sense the channel cooperatively to reduce the sensing errors. The sensing results of individual SUs are assumed to be independent. Let $B$ represent the PU activity such that $B=1$ if PU is active, and $B=0$ otherwise. Let $P_f^i$ denote the {\bf probability of a false alarm} for SU $i$, which is the probability that SU $i$ senses the PU to be active given that the PU is actually idle. $P_m^i$ represents the {\bf probability of mis-detection} for SU $i$, which is the probability that SU $i$ senses the PU to be idle given that the PU is actually active. 

{\bf Cooperative Sensing}: Multiple SUs are chosen to sense the channel and the SU-BS predicts the PU activity by collecting the sensing results from these SUs. We denote the set of SUs that participate in the cooperative sensing as $S_0$, where $|S_0|=k$. Note that $S_0 \subseteq S$. In the cooperative sensing model, we assume the SU-BS collects sensing results from SUs in $S_0$. 

{\bf Cooperative Sensing Indicator}: The observation of the PU activity by SU $i$ is denoted by $o_i$. $o_i=1$ indicates that SU $i$ observes the PU to be active, while $o_i=0$ indicates that SU $i$ observes the PU to be idle. In this paper, our objective is to characterize $S_0$ and estimate the PU activity based on observations from $S_0$ (called the decision rule). The decision rule is denoted as a function $f: \Omega ^{k}\rightarrow \Omega$ where $\Omega=\{0,1\}$. The observations form a vector $\boldsymbol{o}$, where $\boldsymbol{o}\in \Omega ^{k}$ while the decision is denoted by $O$ where $O\in \Omega$. The false alarm probability of cooperative sensing is denoted by $P_f^c=P(O=1|B=0)$. The mis-detection probability of cooperative sensing is denoted by $P_m^c=P(O=0|B=1)$. One time slot is divided into a control slot $T_c$ and a data slot $T_d$ where $T_c+T_d=1$ (Figure~\ref{fig:slot}). In the control slot, the SU-BS collects sensing results from $S_0$ and notifies an SU in $S$ if the cooperative sensing result is ``idle" ($O=0$). If the PU is active (mis-detection), the PU transmission will be collided with the transmission from the SU. The length $T_c$ of the control slot is regarded as the sensing overhead and assumed to be constant throughout the paper. It means that a fixed time period is allocated for cooperative sensing in each slot.  

\begin{figure}[tb]
    \begin{center}
    \setlength{\unitlength}{1in}
    \includegraphics[width=0.3\textwidth]{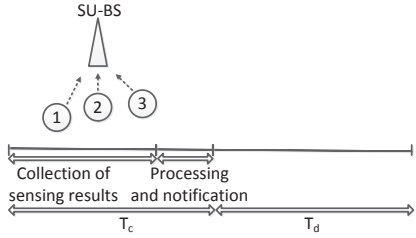}
    \end{center}
\vspace{-1.2em}
\caption{Control slot $T_c$ and data slot $T_d$.}
\vspace{-1.8em}
\label{fig:slot}
\end{figure}   

The uplinks of SUs in $S$ are assumed to have the same capacity which is normalized to $1$. If the decision of the cooperative sensing at the SU-BS is ``idle", the SU-BS notifies one of the SUs in $S$ (not limited to $S_0$, the sensing set) to transmit. We assume SUs in $S$ are always backlogged. The scheduling of the transmitting SUs is beyond the scope of this paper. However, any work-conserving scheduling policy operating on idle slots can be used together with the decision rule to maximize the total system throughput. We let ${\pi}_0$ denote the probability that the PU is idle and we assume that the prior distribution of PU activity is acquired over time accurately. Note that we do not restrict the PU activity to any specific distribution except that it does not change within one time slot. The average throughput of PUs whose transmission would be interfered by SUs in $S$ is denoted as $\gamma$. Table~\ref{tab:not} summarizes the notations used in the paper.
\begin{table}[t]
\centering
\caption{{\scriptsize Notation List}}
\vspace{-1em}
\begin{tabular}{|c|l|}
\hline
{\scriptsize Symbol}&{\footnotesize Meaning}\\\hline\hline
{\scriptsize $M$}&{\scriptsize Total number of SUs in the secondary network}\\\hline
{\scriptsize $S$}&{\scriptsize Set of SUs which cause interference to PU receivers}\\\hline
{\scriptsize $N$}&{\scriptsize $|S|$}\\\hline
{\scriptsize $S_0$}&{\scriptsize Set of SUs that are chosen to sense the channel. $S_0\subseteq S$}\\\hline
{\scriptsize $k$}&{\scriptsize $|S_0|$}\\\hline
{\scriptsize $P_f^i$}&{\scriptsize False alarm probability of SU $i$}\\\hline
{\scriptsize $P_m^i$}&{\scriptsize Mis-detection probability of SU $i$}\\\hline
{\scriptsize $P_f^c$}&{\scriptsize False alarm probability of cooperative sensing}\\\hline
{\scriptsize $P_m^c$}&{\scriptsize Mis-detection probability of cooperative sensing}\\\hline
{\scriptsize $T_c$}&{\scriptsize Control slot}\\\hline
{\scriptsize $T_d$}&{\scriptsize Data slot}\\\hline
{\scriptsize ${\pi}_0$}&{\scriptsize Probability that the PU is idle}\\\hline
{\scriptsize $\gamma$}&{\scriptsize Average throughput of PUs in the interference range of a SU}\\\hline
\end{tabular}
\label{tab:not}
\vspace{-1em}
\end{table}

The SU communication follows a protocol with the following outline:

1) SUs report $P_m^i$'s and $P_f^i$'s to SU-BS;

2) SU-BS determines the sensing set $S_0$ and the decision rule $f$ based on $P_m^i$, $P_f^i$'s and the optimization metric;

3) SU-BS notifies SUs in $S_0$ with an $ACK$ and also assigns each one of them a $SEQ$ number for reporting sensing results;

4) SUs receiving an $ACK$ sense the channel and report the results to SU-BS in the order of $SEQ$;

5) SU-BS makes the decision of the PU activity based on the sensing results and $f$ and schedules an SU for transmission if the decision is $0$ (PU idle).

\section{System Throughput Maximization}
\label{sec:max}
In this section, we formulate the cooperative sensing problem with the assumption that $S_0=S$, that is, the sensing results from all SUs in $S$ are reported to SU-BS within $T_c$. SUs outside $S$ can transmit without causing interference to the PUs. Thus, their performance is independent of the choice of the sensing set or the decision rule. Our goal is to maximize the sum of the expected throughput of SUs in $S$ and that of the PUs whose transmission may be interfered by the SUs. It is equivalent to maximizing the expected throughput of the system with PU-SU co-existence. 

\subsection{Problem Formulation}
\label{subsec:max_formulation}

Given $B=0$ (the PU is idle), the probability of a particular observation vector $\boldsymbol{o}$ occurring is 
\vspace{-0.5em}
\begin{equation}
P(\boldsymbol{o}|B=0)=\prod\limits_{i\in S,o_i=1}{P_f^i} \prod\limits_{j\in S,o_j=0}{(1-P_f^j)}.
\vspace{-0.5em}
\end{equation}

The sum of all $P(\boldsymbol{o}|B=0)$'s with $f(\boldsymbol{o})=0$ is given as 
\vspace{-0.2em}

\begin{equation}
P(O=0|B=0)=\sum\limits_{f(\boldsymbol{o})=0}{P(\boldsymbol{o}|B=0)}.
\vspace{-0.5em}
\label{eq:correct0}
\end{equation}
Then, the false alarm probability of cooperative sensing is
\vspace{-0.5em}
\begin{equation} 
P_f^c=1-P(O=0|B=0)=1-\sum\limits_{f(\boldsymbol{o})=0}{P(\boldsymbol{o}|B=0)}.
\vspace{-0.5em}
\label{eq:bin0prob}
\end{equation}

Likewise, given $B=1$ (the PU is active), the probability of a particular observation vector $\boldsymbol{o}$ occurring is 
\vspace{-0.2em}
\begin{equation}
P(\boldsymbol{o}|B=1)=\prod\limits_{i\in S,o_i=1}{(1-P_m^i)} \prod\limits_{j\in S,o_j=0}{P_m^j}.
\vspace{-0.5em}
\end{equation}

The sum of all $P(\boldsymbol{o}|B=1)$'s with $f(\boldsymbol{o})=1$ is given as 
\vspace{-0.5em}
\begin{equation}
P(O=1|B=1)=\sum\limits_{f(\boldsymbol{o})=1}{P(\boldsymbol{o}|B=1)}.
\vspace{-0.5em}
\label{eq:correct1}
\end{equation}

Then, the mis-detection probability of cooperative sensing is 
\vspace{-0.5em}
\begin{equation} 
P_m^c=1-P(O=1|B=1)=1-\sum\limits_{f(\boldsymbol{o})=1}{P(\boldsymbol{o}|B=1)}.
\vspace{-0.5em}
\label{eq:bin1prob}
\end{equation}

Note that Equation~(\ref{eq:correct0}) is the conditional probability that SU-BS correctly identifies the PU activity when it is idle so that one SU could transmit successfully; Equation~(\ref{eq:correct1}) is the conditional probability that SU-BS correctly detects the PU is active so that no SU would transmit and the PU could transmit successfully. Accordingly, the expected throughput of the SUs can be represented by
\vspace{-0.5em}
\[
(1-T_c)P(B=0,O=0)=(1-T_c){\pi}_0 P(O=0|B=0)
\vspace{-0.2em}\]
\begin{equation}
=(1-T_c){\pi}_0 \sum\limits_{f(\boldsymbol{o})=0}{P(\boldsymbol{o}|B=0)},\vspace{-0.3em}
\end{equation}
since the uplinks of SUs in $S$ are assumed to have capacity $1$ and only one of them could be scheduled in each time slot. The expected throughput of the PU can be represented by
\vspace{-0.5em}
\begin{equation}
\gamma P(O=1|B=1)=\gamma \sum\limits_{f(\boldsymbol{o})=1}{P(\boldsymbol{o}|B=1)}
\vspace{-0.5em}
\end{equation}
since $\gamma$ is the average throughput of the PU whose transmission would be interfered by SUs in $S$. The problem is then formulated as follows:

Problem (A):
\[\max\limits_{f}{(1-T_c){\pi}_0 \sum\limits_{f(\boldsymbol{o})=0}{P(\boldsymbol{o}|B=0)}+\gamma \sum\limits_{f(\boldsymbol{o})=1}{P(\boldsymbol{o}|B=1)}}\]

\subsection{Optimal Solution with Bayesian Decision Rule}
\label{subsec:Bayesian_max}

We show that Problem (A) can be converted to a Bayesian Decision problem. Algorithm~\ref{alg:Bayesian} is then proposed based on Bayesian decision rule to minimize the posterior expected loss \cite{Berger} and it is of constant time complexity.

\begin{algorithm}[t]
    \caption{{\footnotesize Bayesian Decision Rule Based Algorithm for maximizing the system throughput (given $\boldsymbol{o}$, decide $O$)}}\label{alg:Bayesian}
    \begin{algorithmic}[1]
	{\footnotesize \IF{$(1-T_c){\pi}_0 \prod\limits_{o_i=1}{P_f^i}\prod\limits_{o_j=0}{(1-P_f^j)}\ge \gamma \prod\limits_{o_i=1}{(1-P_m^i)}\prod\limits_{o_j=0}{P_m^j}$}
		\STATE $O\leftarrow 0$
	\ELSE
		\STATE $O\leftarrow 1$
	\ENDIF}
    \end{algorithmic}
    \vspace{-0.2em}
\end{algorithm}

Problem (A) is equivalent to Problem~(\ref{eq:min}) in terms of optimal $f$. \vspace{-1em}
\[\max\limits_{f}{L(B=0,O=1)\left[{\pi}_0\sum\limits_{f(\boldsymbol{o})=1}{P(\boldsymbol{o}|B=0)}\right]}\]
\vspace{-0.5em}
\begin{equation}
+L(B=1,O=0) \left[(1-{\pi}_0)\sum\limits_{f(\boldsymbol{o})=0}{P(\boldsymbol{o}|B=1)}\right],
\label{eq:min}
\end{equation}
where $L(B,O)$ is the loss of decision $O$ based on observation $\boldsymbol{o}$, which is a negative number. $L(B=0,O=1)=-(1-T_c)$ and $L(B=1,O=0)=-\frac{\gamma}{1-{\pi}_0}$. Thus Equation~(\ref{eq:min}) is the \emph{posterior expected loss} of decision $O$ (Definition $8$ of Chapter 4.4 in \cite{Berger}). Using the Bayesian decision rule, Problem~(\ref{eq:min}) can be solved optimally \cite{Berger}: given $\boldsymbol{o}$, the decision $O=1$ if $|L(B=0,O=1)|{\pi}_0 P(\boldsymbol{o}|B=0)<|L(B=1,O=0)|(1-{\pi}_0)P(\boldsymbol{o}|B=1)$ and $O=0$ otherwise. Algorithm~\ref{alg:Bayesian} is designed accordingly.   

\section{Guaranteeing a Target PU Throughput}
\label{sec:PU_thru}
In this section, we investigate the maximum throughput problem with a PU throughput constraint. With higher priority, a minimum PU throughput is guaranteed in the problem formulation. We first show that this constrained problem is NP-hard by reducing the classical partition problem \cite{Garey} to it. Then a greedy approximation algorithm is proposed to achieve strictly greater than $1/2$ of the optimal solution. The complexity of the algorithm is shown to be pseudo-polynomial by solving a two-dimensional dynamic programming problem. 

\subsection{Problem Formulation and Properties}
\label{subsec:form}
We formulate the constrained optimization problem as follows:

Problem (B):
\vspace{-0.5em}
\[\max\limits_{f}{(1-T_c){\pi}_0 \sum\limits_{f(\boldsymbol{o})=0}{P(\boldsymbol{o}|B=0)}+\gamma \sum\limits_{f(\boldsymbol{o})=1}{P(\boldsymbol{o}|B=1)}}\]
\vspace{-0.8em}
\begin{equation}
\mbox{s.t. }\sum\limits_{f(\boldsymbol{o})=1}{P(\boldsymbol{o}|B=1)}\ge \alpha .
\vspace{-0.5em}
\label{eq:PUcons} 
\end{equation}
Equation~(\ref{eq:PUcons}) is the constraint we put on Problem~(B) where the expected PU throughput must be no less than a preset system-dependent threshold. It is equivalent to $1-\sum\limits_{f(\boldsymbol{o})=1}{\prod\limits_{i\in S,o_i=1}{(1-P_m^i)}\prod\limits_{j\in S,o_j=0}{P_m^j}}\le 1-\alpha $, where $1-\alpha$ is the collision factor. This can be interpreted as the probability that a PU transmission colliding with an SU transmission being no greater than $1-\alpha$. Problem (B) maximizes the expected system throughput given that the lowest PU throughput can be met considering the high priority of the PU in cognitive radio networks.

By observing the structure of Problem (B), we state Lemma~\ref{lem:Bayesian} that shows the optimal assignment of observations with $G(\boldsymbol{o})<H(\boldsymbol{o})$ where 
$G(\boldsymbol{o})=(1-T_c){\pi}_0 P(\boldsymbol{o}|B=0)$ and 
$H(\boldsymbol{o})=\gamma P(\boldsymbol{o}|B=1)$. We define $f^*$ as the optimal solution to Problem~(B).

\begin{lemma}
\label{lem:Bayesian}
In the optimal solution to Problem (B), we have $f^*(\boldsymbol{o})=1$ for all $G(\boldsymbol{o})<H(\boldsymbol{o})$.
\end{lemma}

\begin{proof}
(Prove by contradiction) Assume that $f^*(\boldsymbol{o})=0$ for some $\boldsymbol{o}$ where $G(\boldsymbol{o})<H(\boldsymbol{o})$. Moving it from $O=0$ to $O=1$ increases $\sum\limits_{f(\boldsymbol{o})=1}{P(\boldsymbol{o}|B=1)}$ so that this operation still makes a feasible solution. Furthermore, the expected system throughput increases considering $G(\boldsymbol{o})<H(\boldsymbol{o})$, which makes a better solution than the current optimal one. It causes a contradiction. Hence, we have $f^*(\boldsymbol{o})=1$ for all $G(\boldsymbol{o})<H(\boldsymbol{o})$ in the optimal solution to Problem (B).
\end{proof}

With the property of Lemma~\ref{lem:Bayesian}, we only need to decide which observations with $G(\boldsymbol{o})\ge H(\boldsymbol{o})$ should be put in $O=1$ to solve Problem (B) optimally. We define $\chi=\{\boldsymbol{o}:G(\boldsymbol{o})\ge H(\boldsymbol{o})\mbox{ and }f^*(\boldsymbol{o})=1\}$, which is the set of observations that need to be moved to $O=1$ in the optimal solution; $\psi=\{\boldsymbol{o}:G(\boldsymbol{o})\ge H(\boldsymbol{o})\mbox{ and }f^*(\boldsymbol{o})=0\}$, which is the set of observations that stay in $O=0$ in the optimal solution. Moreover, we define $A=\sum\limits_{\boldsymbol{o}:G(\boldsymbol{o})<H(\boldsymbol{o})}{H(\boldsymbol{o})}$, which is the contribution of observations with $G(\boldsymbol{o})<H(\boldsymbol{o})$ in the optimal solution; $B=\sum\limits_{\boldsymbol{o}\in \chi}{G(\boldsymbol{o})}$, which is the contribution of observations in $\chi$ when put in $O=0$; $B'=\sum\limits_{\boldsymbol{o}\in \chi}{H(\boldsymbol{o})}$, which is the contribution of observations in $\chi$ when put in $O=1$ ($B\ge B'$); $C=\sum\limits_{\boldsymbol{o}\in \psi}{G(\boldsymbol{o})}$, which is the contribution of observations in $\psi$ when put in $O=0$; $C'=\sum\limits_{\boldsymbol{o}\in \psi}{H(\boldsymbol{o})}$, which is the contribution of observations in $\psi$ when put in $O=1$ ($C \ge C'$). Then, $A+B+C$ is the optimal solution to Problem (B) without the PU throughput constraint; $A+B'+C$ is the optimal solution to Problem (B), which is no greater than $A+B+C$. The optimal assignment is illustrated in Figure~\ref{fig:bin}. 

\begin{figure}[tb]
    \begin{center}
    \setlength{\unitlength}{1in}
    \includegraphics[scale=0.3]{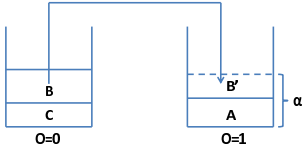}
    \end{center}
\vspace{-1.2em}
\caption{Assignment of observations with no constraint and the optimal assignment for Problem (B). $A+B+C$ is the optimal assignment with no constraint while $A+B'+C$ is the optimal solution to Problem (B).}
\label{fig:bin}
\vspace{-1.5em}
\end{figure} 

\subsection{Proof of NP-hardness}
\label{subsec:NP}

We focus on identifying the hardness of deciding observations with $G(\boldsymbol{o})\ge H(\boldsymbol{o})$ that should be put in $O=1$ in this section. It is shown to be NP-hard by reducing the classical partition problem \cite{Garey} to the subproblem of it in Theorem~\ref{thm:NP}. 

\begin{theorem}
Problem~(B) is NP-hard.
\label{thm:NP}
\end{theorem}

\begin{proof}
We first state the classical partition problem \cite{Garey} - Given $N$ positive integers: $y_1$, $\cdots$, $y_{N}$, is there a way to have them partitioned into two equal-sized subsets that have the same sum? For reduction, we construct an instance of Problem (B) by setting $(1-T_c){\pi}_0=\gamma$, $P_m^i+P_f^i<1$, $\alpha=\epsilon+\sum\limits_{\boldsymbol{o}:G(\boldsymbol{o})<H(\boldsymbol{o})}{H(\boldsymbol{o})}$ with $\epsilon\le \min\limits_{\boldsymbol{o}:G(\boldsymbol{o})\ge H(\boldsymbol{o})}{H(\boldsymbol{o})}$. For this instance, putting any $\boldsymbol{o}$ with $G(\boldsymbol{o})\ge H(\boldsymbol{o})$ to $O=1$ would make a feasible solution given that observations with $G(\boldsymbol{o})<H(\boldsymbol{o})$ have all been put in $O=1$. Obviously, choosing the observation with minimum non-negative $G(\boldsymbol{o})-H(\boldsymbol{o})$ would be the optimal solution. Note that $G(\boldsymbol{o})-H(\boldsymbol{o})=0$ is equivalent to $\log{\frac{G(\boldsymbol{o})}{H(\boldsymbol{o})}}=0$. By setting $\log{\frac{1-P_f^i}{P_m^i}}=-\log{\frac{P_f^i}{1-P_m^i}}=y_i$ for all $i$, we have $\log{\frac{G(\boldsymbol{o})}{H(\boldsymbol{o})}}=\sum\limits_{o_i=0,i=1,\cdots,N}{y_i}+\sum\limits_{o_j=1,j=1,\cdots,N}{-y_j}$. Now the instance becomes: given $N$ pairs of integers $(y_1,-y_1),\cdots,(y_N,-y_N)$, exactly one number should be chosen from each pair; with this constraint, what is the minimum non-negative sum? The reduction from the partition problem to this instance of Problem (B) can be done in polynomial time.

To verify the correctness of the reduction, we can check: if the minimum non-negative $G(\boldsymbol{o})-H(\boldsymbol{o})$ is $0$, that is, the optimal solution of the instance is $0$, we can answer ``Yes" to the partition problem; if it is positive, we can answer ``No" to the partition problem. If Problem (B) can be solved in polynomial time, then the partition problem can be solved in polynomial time as well. The partition problem is well-known to be NP-complete \cite{Garey}. Assuming P$\neq$NP, Problem (B) has been proven to be NP-hard.
\end{proof}

It has been shown in Theorem~\ref{thm:NP} that finding the observation with $G(\boldsymbol{o})$ closest to $H(\boldsymbol{o})$ from above is NP-hard. Hence, it is unlikely to find an efficient algorithm to solve Problem~(B) optimally. We will focus on the approximation algorithm design in Section~\ref{subsec:greedy}.  

\subsection{Greedy Approximation Algorithm}
\label{subsec:greedy}

We propose a greedy algorithm (Algorithm~\ref{alg:greedy}) that initially assigns all observations to $O=1$ and then moves observations with $G(\boldsymbol{o})\ge H(\boldsymbol{o})$ by $\frac{G(\boldsymbol{o})}{H(\boldsymbol{o})}$ from the highest to lowest to $O=0$ until the feasibility constraint of Problem (B) is violated. By transforming Problem~(B) into the Knapsack Problem \cite{Vijay}, we will show that the algorithm achieves strictly greater than $1/2$ of the optimal solution for Problem (B), which is $A+B'+C$ in Figure~\ref{fig:bin}. Although the sum of $G(\boldsymbol{o})$ or $H(\boldsymbol{o})$ in the worst case has exponential number of terms, we will design a pseudo-polynomial time algorithm in Section~\ref{subsec:comp} considering its combinatorial nature. Ignoring rounding errors, the implementation calculates these sums accurately.

\begin{algorithm}[t]
    \caption{{\footnotesize Greedy Approximation Algorithm for Problem (B)}}\label{alg:greedy}
    {\scriptsize Input: $N$, $T_c$, ${\pi}_0$, $\gamma$, $\alpha$, $P_m^i$, $P_f^i$ for all $i$ \\Output: $f$ or ``infeasible"}
    \begin{algorithmic}[1]
    {\footnotesize \STATE $G(\boldsymbol{o})\leftarrow (1-T_c){\pi}_0 \prod\limits_{i\in S,o_i=1}{P_f^i}\prod\limits_{j\in S,o_j=0}{(1-P_f^j)}$ for all $\boldsymbol{o}$
	\STATE $H(\boldsymbol{o})\leftarrow \gamma \prod\limits_{i\in S,o_i=1}{(1-P_m^i)}\prod\limits_{j\in S,o_j=0}{P_m^j}$ for all $\boldsymbol{o}$
	\STATE $U\leftarrow \sum\limits_{\boldsymbol{o}}{H(\boldsymbol{o})}$
	\IF{$U<\alpha \times \gamma$}
	\STATE output ``infeasible" and return
    \ENDIF
    \STATE $f(\boldsymbol{o})\leftarrow 1$ for all $\boldsymbol{o}$, $sum1 \leftarrow \sum\limits_{\boldsymbol{o}:G(\boldsymbol{o})<H(\boldsymbol{o})}{H(\boldsymbol{o})}$
    \IF{$sum1\ge \alpha \times \gamma$}
		\STATE $f(\boldsymbol{o})=0$ for all $\boldsymbol{o}$ with $G(\boldsymbol{o})\ge H(\boldsymbol{o})$ and return
    \ENDIF
    \STATE Sort $\boldsymbol{o}$'s with $G(\boldsymbol{o})\ge H(\boldsymbol{o})$ in non-increasing order of $\frac{G(\boldsymbol{o})}{H(\boldsymbol{o})}$ and label them from $1$ to $l$
    \STATE $sum2\leftarrow 0$

    \FOR{$i=1$ to $l$}
		\STATE \textbf{if} $sum2 +H(\boldsymbol{o}_i)> U-\alpha\times \gamma$ \textbf{then} break

		\STATE $sum2\leftarrow sum2+H(\boldsymbol{o}_i)$, $f(\boldsymbol{o}_i)\leftarrow 0$

    \ENDFOR}

    \end{algorithmic}
\end{algorithm}

In Algorithm~\ref{alg:greedy}, observations are chosen by $\frac{G(\boldsymbol{o})}{H(\boldsymbol{o})}$ from the highest to the lowest and assigned to $O=0$ after those with $G(\boldsymbol{o})< H(\boldsymbol{o})$ are assigned to $O=1$. Ties are broken by putting observations with smaller $H(\boldsymbol{o})$ in the front. In Line~$3$, $U$ is assigned to be the sum of contributions of all observations if put in $O=1$ ($A+B'+C'$). In Lines~$4$-$5$, whether a feasible solution exists for the given input is checked by comparing the extreme case where all observations are assigned to $O=1$ with the threshold $\alpha \times \gamma$. In Line~$6$, observations are initialized to $O=1$. Lines~$6$-$8$ checks whether the feasibility constraint in Problem (B) has been satisfied under the initial assignment. If yes, observations with $G(\boldsymbol{o})\ge H(\boldsymbol{o})$ are assigned to $O=0$ by Bayesian decision rule. Lines~$9$-$13$ searches for observations with $G(\boldsymbol{o})\ge H(\boldsymbol{o})$ from the highest $\frac{G(\boldsymbol{o})}{H(\boldsymbol{o})}$ to lowest until $sum2 +H(\boldsymbol{o}_i)\le U-\alpha  \times \gamma$ is violated (Line~$12$). Note that $\sum\limits_{\boldsymbol{o}:f(\boldsymbol{o})=0}{H(\boldsymbol{o})}\le U-\alpha \times \gamma$ and $\sum\limits_{\boldsymbol{o}:f(\boldsymbol{o})=1}{H(\boldsymbol{o})}\ge \alpha \times \gamma$ (feasibility constraint) are equivalent since $\sum\limits_{\boldsymbol{o}}{H(\boldsymbol{o})}=U$. $f(\boldsymbol{o})$ of these observations are set to be $0$ (Line $13$) in the searching process. Next, we state Theorem~\ref{thm:greedy} that gives the approximation factor of Algorithm~\ref{alg:greedy}.

\begin{theorem}
\label{thm:greedy}
Algorithm~\ref{alg:greedy} achieves strictly greater than $1/2$ of the optimal solution to Problem (B).
\end{theorem}

\begin{proof}
Recall the following notations: $A=\sum\limits_{\boldsymbol{o}:G(\boldsymbol{o})<H(\boldsymbol{o})}{H(\boldsymbol{o})}$, $B=\sum\limits_{\boldsymbol{o}\in \chi}{G(\boldsymbol{o})}$, $B'=\sum\limits_{\boldsymbol{o}\in \chi}{H(\boldsymbol{o})}$, $C=\sum\limits_{\boldsymbol{o}\in \psi}{G(\boldsymbol{o})}$, $C'=\sum\limits_{\boldsymbol{o}\in \psi}{H(\boldsymbol{o})}$, and $W=A+B+C$. We define $U=A+B'+C'=\sum\limits_{\boldsymbol{o}}{H(\boldsymbol{o})}$. Let $APX$ be the solution to Problem (B) output by Algorithm~\ref{alg:greedy}. Let $OPT$ be the optimal solution to Problem (B). Then, we have 
\vspace{-0.5em}
\begin{equation}
\vspace{-0.5em}
OPT=W-(B-B')
\label{eq:opt}
\end{equation}
\[= A+C+B'=(A+B'+C')+(C-C')=U+(C-C').
\vspace{-0.3em}
\]
Equation~(\ref{eq:opt}) holds by Figure~\ref{fig:bin} since moving observations from $O=0$ to $O=1$ loses $B-B'$ in throughput compared to the optimal solution with no constraint, which is $W$. Note that $(C-C')$ is the optimal solution to Problem~(\ref{eq:cons}) and it is the difference of contribution to throughput between keeping observations with $G(\boldsymbol{o})\ge H(\boldsymbol{o})$ and $f^*(\boldsymbol{o})=0$ in $O=0$ and moving them to $O=1$. 
\[\vspace{-0.5em}
\max \sum\limits_{G({\boldsymbol{o}}_i)\ge H({\boldsymbol{o}}_i),i=1,\cdots, l}{(G({\boldsymbol{o}}_i)-H({\boldsymbol{o}}_i)) x_i}\]
\begin{equation}
\mbox{s.t. }\sum\limits_{G({\boldsymbol{o}}_i) \ge H({\boldsymbol{o}}_i),i=1,\cdots, l}{H({\boldsymbol{o}}_i) x_i}\le U-\alpha \times \gamma
\label{eq:cons}
\end{equation}
\[\vspace{-0.3em}
x_i \in \{0,1\}\mbox{ for all }i,\]
where observations with $G(\boldsymbol{o})\ge H(\boldsymbol{o})$ are labeled in an arbitrary order. Next, we will show that the constraint of Problem~(\ref{eq:cons}) and that of Problem (B) are equivalent. As shown in Figure~\ref{fig:bin}, we have 
\vspace{-0.5em}
\[
\vspace{-0.2em}
A+B'\ge \alpha \times \gamma \Leftrightarrow A+B'-\alpha \times \gamma\ge 0\]
\begin{equation}
\Leftrightarrow A+B'+C'-\alpha\times \gamma\ge C' \Leftrightarrow C'\le U-\alpha\times \gamma .
\vspace{-0.3em}
\end{equation}

Problem~(\ref{eq:cons}) is a Knapsack Problem \cite{Vijay}. By the following greedy approach, at least $1/2$ of $(C-C')$ can be achieved \cite{Vijay}: choosing observations with $G({\boldsymbol{o}}_i)\ge H({\boldsymbol{o}}_i)$ from the highest $\frac{G({\boldsymbol{o}}_i)-H({\boldsymbol{o}}_i)}{H({\boldsymbol{o}}_i)}$ to the lowest until (\ref{eq:cons}) is violated, which is exactly what we do in Algorithm~\ref{alg:greedy} since $\frac{G({\boldsymbol{o}}_i)-H({\boldsymbol{o}}_i)}{H({\boldsymbol{o}}_i)}\ge \frac{G({\boldsymbol{o}}_j)-H({\boldsymbol{o}}_i)}{H({\boldsymbol{o}}_j)}$ if and only if $\frac{G({\boldsymbol{o}}_i)}{H({\boldsymbol{o}}_i)}\ge \frac{G({\boldsymbol{o}}_j)}{H({\boldsymbol{o}}_j)}$. Hence, $APX=U+1/2(C-C')$ holds. Since $U>0$, we always have $APX/OPT>1/2$ for Problem (B). 
\end{proof}

So far, we have shown that the greedy approximation algorithm (Algorithm~\ref{alg:greedy}) exists for Problem (B) with an approximation factor strictly greater than $1/2$. When $U\gg C-C'$, this factor could be arbitrarily close to $1$.

\subsection{Pseudo-Polynomial Implementation}
\label{subsec:comp}
In Lines~$3$, $6$, $9$ and $12$ of Algorithm~\ref{alg:greedy}, exponential number of observations are involved in the worst case due to its combinatorial nature. We design a pseudo-polynomial time algorithm by means of dynamic programming for the implementation. The running time of a pseudo-polynomial time algorithm is polynomial in the numeric value of the input, which is exponential in the length of them assuming they are rational numbers \cite{Vijay}. For simplicity, we assume $(1-T_c){\pi}_0=\gamma$ in Algorithm~\ref{alg:pseudo}, which can though be extended to general cases without the assumption easily.

\begin{algorithm}[t]
    \caption{{\footnotesize Pseudo-Polynomial Algorithm to Find the Joint Distribution of $(\log{\frac{G(\boldsymbol{o})}{H(\boldsymbol{o})}},\log{H(\boldsymbol{o})})$}}\label{alg:pseudo}
    {\scriptsize Input: $N$, $P_f^i$, $P_m^i$ for all $i$ \\Output: $C(N,j,j')$ for all $j$, $j'$}
    \begin{algorithmic}[1]
    {\footnotesize \STATE $y_i\leftarrow round(\log{\frac{1-P_f^i}{P_m^i}},r)\times 10^r$ for all $i$
	\STATE $z_i\leftarrow round(\log{\frac{P_f^i}{1-P_m^i}},r)\times 10^r$ for all $i$
	\STATE ${\lambda}_i\leftarrow round(\log{P_m^i},r)\times 10^r$ for all $i$
	\STATE ${\mu}_i\leftarrow round(\log{(1-P_m^i)},r)\times 10^r$ for all $i$
	\STATE $M\leftarrow \sum\limits_{i=1}^{N}{\max{\{y_i,z_i\}}}$, $m\leftarrow \sum\limits_{i=1}^{N}{\min{\{y_i,z_i\}}}$
	\STATE $M'\leftarrow\max{\{\max\limits_{i}{{\lambda}_i},\max\limits_{i}{{\mu}_i}\}}$, $m'\leftarrow\sum\limits_{i=1}^{N}{\min{\{{\lambda}_i,{\mu}_i\}}}$
	\STATE $C(i,j,j')\leftarrow 0$ for all $i$, $j$, $j'$, $C(1,y_1,{\lambda}_1)\leftarrow 1$, $C(1,z_1,{\mu}_1)\leftarrow 1$
	\FOR{$i=1$ to $N-1$}
		\FOR{$j=m$ to $M$}
			\STATE \textbf{for} $j'=m'$ to $M'$ \textbf{do} $C(i+1,j,j')=C(i,j-y_{i+1},j'-{\lambda}_{i+1})+C(i,j-z_{i+1},j'-{\mu}_{i+1})$
			
		\ENDFOR
	\ENDFOR}
    \end{algorithmic}
\end{algorithm}

In Algorithm~\ref{alg:pseudo}, dynamic programming is applied to calculate the joint distribution of $\log{\frac{G(\boldsymbol{o})}{H(\boldsymbol{o})}}$ and $\log{H(\boldsymbol{o})}$, which counts the number of observations with the same $\log{\frac{G(\boldsymbol{o})}{H(\boldsymbol{o})}}$ and the same $\log{H(\boldsymbol{o})}$. Note that this algorithm does not require future information - only the collection of all sensing results from SUs in the current time slot is required. Lines~$3$, $6$, $9$ and $12$ of Algorithm~\ref{alg:greedy} can be calculated based on these counts. $round(a,r)$ rounds $a$ to $r$ decimal places. We use $round(a,r)\times 10^r$ to scale and round a real $a$ to an integer. The rounding error will be discussed in the simulation. $M$ and $m$ specify the maximum and minimum contribution an observation $\boldsymbol{o}$ can have to $\log{\frac{G(\boldsymbol{o})}{H(\boldsymbol{o})}}$ respectively, while $M'$ and $m'$ specify the maximum and minimum contribution an observation $\boldsymbol{o}$ can have to $\log{H(\boldsymbol{o})}$ respectively. $C(N,j,j')$ records the number of observations with $\log{\frac{G(\boldsymbol{o})}{H(\boldsymbol{o})}}$ (after rounding) equal to $j$ and $\log{H(\boldsymbol{o})}$ (after rounding) equal to $j'$. Boundary conditions are set in Line~$7$. Lines~$8$-$10$ use iterations to find $C(i,j,j')$ for all $i=1,\cdots,N$, $m\le j\le M$ and $m'\le j'\le M'$. The recursive function in Line~$10$ matches the fact that when the $0$ observation from SU $i+1$ is added to observations from SU $1$ to $i$, $\log{\frac{G(\boldsymbol{o})}{H(\boldsymbol{o})}}$ is added by $y_{i+1}$ and $\log{H(\boldsymbol{o})}$ is added by ${\lambda}_{i+1}$; on the other hand, when the $1$ observation from SU $i+1$ is added to observations from SU $1$ to $i$, $\log{\frac{G(\boldsymbol{o})}{H(\boldsymbol{o})}}$ is added by $z_{i+1}$ and $\log{H(\boldsymbol{o})}$ is added by ${\mu}_{i+1}$. Note that Line~$10$ may encounter $C(i,j,j')$ beyond the boundaries of $j$ or $j'$, the value of which will be treated as $0$. The time complexity is $O(N(M-m)(M'-m'))$, which is pseudo-polynomial.

After the $C(N,j,j')$ distribution is found, Lines~$3$, $6$, $9$ and $12$ of Algorithm~\ref{alg:greedy} can be calculated accordingly, the time complexity of which is dominated by $O(N(M-m)(M'-m'))$. 

\section{Sensing Set Identification}
\label{sec:uncertain_set}
In this section, we formulate a new problem where the SU-BS is free to choose any subset of $S$ as the sensing set and maximizes the expected throughput of the system. We define $d$ as the homogeneous reporting delay of the sensing results from an SU to the SU-BS, and $\tilde{d}$ as the miscellaneous delay which covers all the processing required after the collection of sensing results at the SU-BS in $T_c$. No matter how many SUs are chosen in the sensing set, we always allocate the length of $T_c$ as the control slot. Thus, the control overhead is still a constant in this section. Two cases are considered: $Nd+\tilde{d}\le T_c$, where all SUs are allowed in the sensing set (Section~\ref{subsec:const}); $Nd+\tilde{d}\ge T_c$, where at most $k=\lfloor \frac{T_c-\tilde{d}}{d} \rfloor$ SUs are allowed in the sensing set (Section~\ref{subsec:exhaustive}). We show that the system throughput is monotonic over the number of SUs chosen in $S_0$ in Proposition~\ref{prop:more_better}. The hardness of the constrained problem with $Nd+\tilde{d}\ge T_c$ is unknown and there is no efficient algorithm proposed for this type of problem so far. 

\subsection{$Nd+\tilde{d}\le T_c$}
\label{subsec:const}

Without the constraint of the number of SUs in $S_0$, we will show that the full set gives the most information. First, we define 
\vspace{-0.5em}
\begin{equation}
P(\boldsymbol{o}^{S_0}|B=0)=\prod\limits_{i\in S_0,o_i=1}{P_f^i}\prod\limits_{j\in S_0,o_j=0}{(1-P_f^j)},\vspace{-0.5em}
\end{equation}
which is the probability of a particular observation vector $\boldsymbol{o}^{S_0}$ where $S_0$ is the sensing set occurring given $B=0$, and 
\vspace{-0.2em}
\begin{equation}
P(\boldsymbol{o}^{S_0}|B=1)=\prod\limits_{i\in S_0,o_i=1}{(1-P_m^i)}\prod\limits_{j\in S_0,o_j=0}{P_m^j},\vspace{-0.5em}
\end{equation} 
which is the probability of a particular observation vector $\boldsymbol{o}^{S_0}$ where $S_0$ is the sensing set occurring given $B=1$. Then we formulate the problem as follows:

Problem (C):
\vspace{-0.8em}

\[\max\limits_{f,S_0}{(1-T_c){\pi}_0 \sum\limits_{f(\boldsymbol{o})=0}{P(\boldsymbol{o}^{S_0}|B=0)}}
\vspace{-0.5em}\]
\[+\gamma \sum\limits_{f(\boldsymbol{o})=1}{P(\boldsymbol{o}^{S_0}|B=1)}.
\vspace{-0.5em}\]


\begin{proposition}
\label{prop:more_better}
Let $F^*(S_0)$ be the optimal solution to Problem~(C) with $S_0$ fixed. Then 
\vspace{-0.3em}
\begin{equation}
F^*(\{i_1,\cdots,i_k,i_{k+1}\}) \ge F^*(\{i_1,\cdots,i_k\}).
\end{equation}
\end{proposition}

\begin{proof}
(sketch) By adding the $k+1$-th SU into the sensing set, we can at least achieve the same system throughput as before by ignoring its observation. The detailed proof can be found in our technical report \cite{Li2011}.
\end{proof}

Using Proposition~\ref{prop:more_better}, we will prove it is the best choice to choose the full set as the sensing set in Corollary~\ref{cor:full_best}.

\begin{corollary}
\label{cor:full_best}
For all $D\subset S$, we have $F^*(S) \ge F^*(D)$.
\end{corollary}

\begin{proof}
Given $D\subset S$, we index the elements in $S\setminus D$ from the smallest to the largest as $l_1,\cdots,l_m$ where $m=|S\setminus D|$ and $0\le m \le N$. By Proposition~\ref{prop:more_better}, we have $F^*(D)\le F^*(D\cup \{l_1\})\le F^*(D\cup \{l_1,l_2\})\le \cdots \le F^*(D\cup \{l_1,\cdots, l_m\})=F^*(S)$.
\end{proof}

By Corollary~\ref{cor:full_best}, Problem~(B) can be solved by first setting $S_0^*=S$ and then applying Algorithm~\ref{alg:Bayesian} to find the optimal decision rule. Note that the time complexity is still $O(1)$.

\subsection{$Nd+\tilde{d}> T_c$}
\label{subsec:exhaustive}

We also investigate the case where the number of SUs in $S_0$ is constrained, and state that it is unlikely to have an efficient algorithm to find the optimal solution. By Proposition~\ref{prop:more_better}, the problem can be formulated as follows:

Problem~(D):
\vspace{-0.8em}

\[\max\limits_{f,S_0}{(1-T_c){\pi}_0 \sum\limits_{f(\boldsymbol{o})=0}{P(\boldsymbol{o}^{S_0}|B=0)}}
\vspace{-0.5em}\]
\[+\gamma \sum\limits_{f(\boldsymbol{o})=1}{P(\boldsymbol{o}^{S_0}|B=1)}
\vspace{-0.5em}\]
\[\mbox{s.t. }|S_0|=k, 
\]
where $k=\lfloor \frac{T_c-\tilde{d}}{d} \rfloor$.

It has been shown in \cite{Pena} that no non-exhaustive search method over the subset of $S$ of size $k$ can always solve it optimally when observations are correlated. For the independent observation problem such as Problem (D), however, it is not clear whether exhaustive search would be necessary as shown in \cite{Van}. Many heuristics such as Sequential Forward Selection (SFS, \cite{SFS}), Sequential Backward Selection (SBS, \cite{SFS}) and their variations \cite{Kudo} have been proposed to solve problems of this type. Although we characterize the monotonic property of system throughput over the number of SUs in the sensing set, the complexity of the problem is not clear in the case when $T_c$ is small, compared to $N$.

\section{Simulations}
\label{sec:simu}
In this section, simulation results are presented for the performance of solutions proposed for Problems~(A), (B), (C) and (D). We first compare the performance of Bayesian decision rule (Algorithm~\ref{alg:Bayesian}), majority, AND and OR policies \cite{4533677} in Section~\ref{subsec:A}. Then the performance of the greedy algorithm for Problem (B) (Algorithm~\ref{alg:greedy}), the random selection and the optimal solution are presented in Section~\ref{subsec:B}. The performance of Sequential Forward Selection (SFS, \cite{SFS}) is compared with the optimal solution to Problem (D) in Section~\ref{subsec:CD}. In all simulation studies, we consider a cognitive radio network with $N=10$, $T_c=0.2$, ${\pi}_0=0.4$, and $\gamma=2$. For each parameter setting, we generate $30$ groups of $P_m^i$'s and $P_f^i$'s randomly, which represents the random geographical locations of SUs in a CRN.

\subsection{Performance of Bayesian Decision rule}
\label{subsec:A}

Algorithm~\ref{alg:Bayesian}, the Bayesian decision rule based algorithm, has been proven to be optimal in Section~\ref{sec:max}. In Figure~\ref{fig:all4}, we demonstrate the increase from majority, AND, OR rules in terms of system throughput, which is the objective function value of Problem (A). In majority rule, the decision is $1$ only when the majority of the SUs sense an active PU; in AND rule, the decision is $1$ only when all SUs sense an active PU; in OR rule, the decision is $1$ if any of the SUs senses an active PU. We vary $\gamma$, the average PU throughput, and $N$, the number of SUs, respectively. Among all four algorithms, Bayesian decision rule strictly outperforms the other three. Among them, the OR rule is better then AND and majority rules since the PU transmission is better protected by the OR rule. OR rule is too conservative to guarantee SU transmission.  

\begin{figure}[!t]
\centering
\subfigure[With different $\gamma$'s.]{
\includegraphics[scale=0.25]{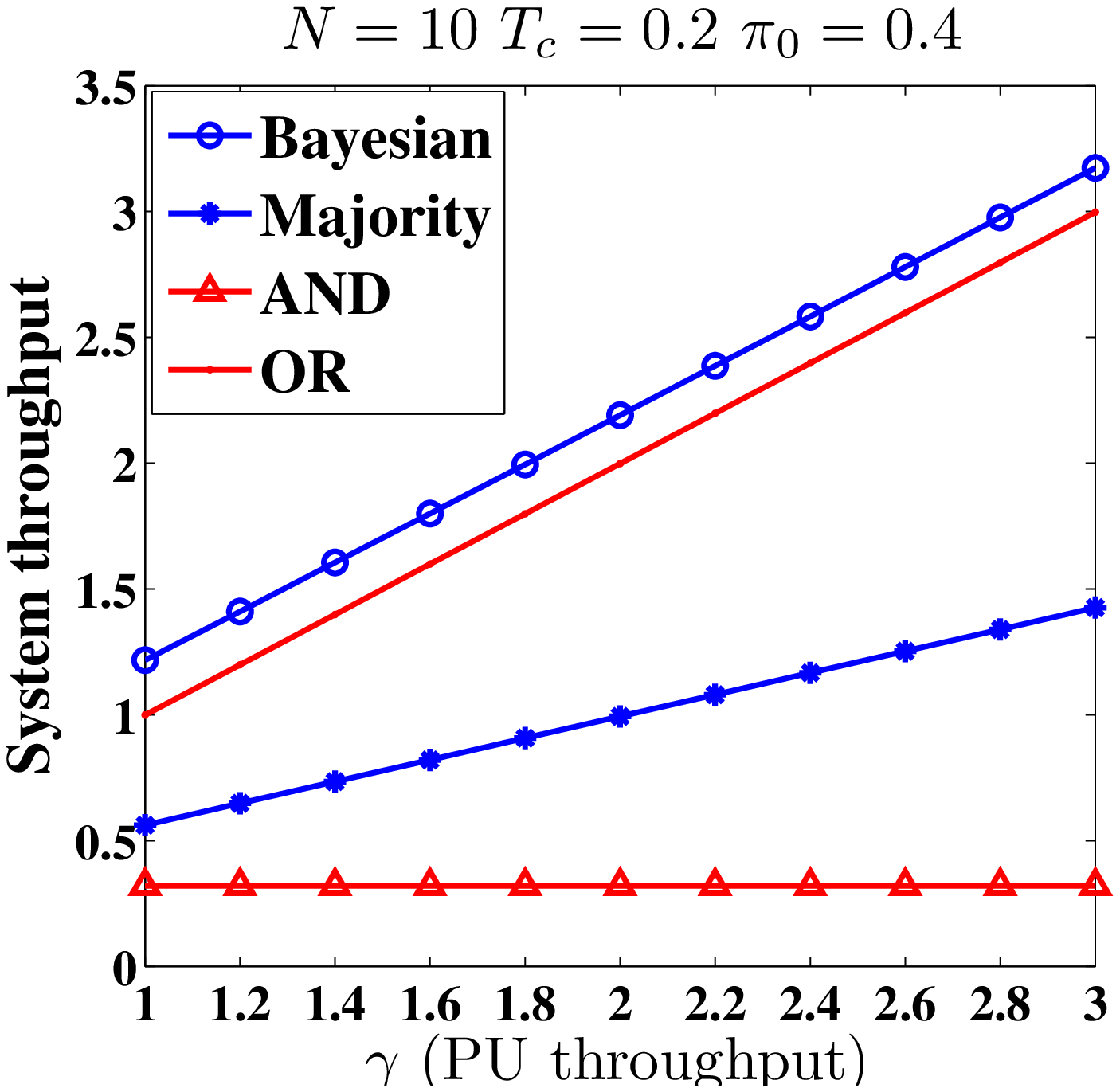}
\label{fig:gamma_B}
}
\hfil
\subfigure[With different $N$'s.]{
\includegraphics[scale=0.25]{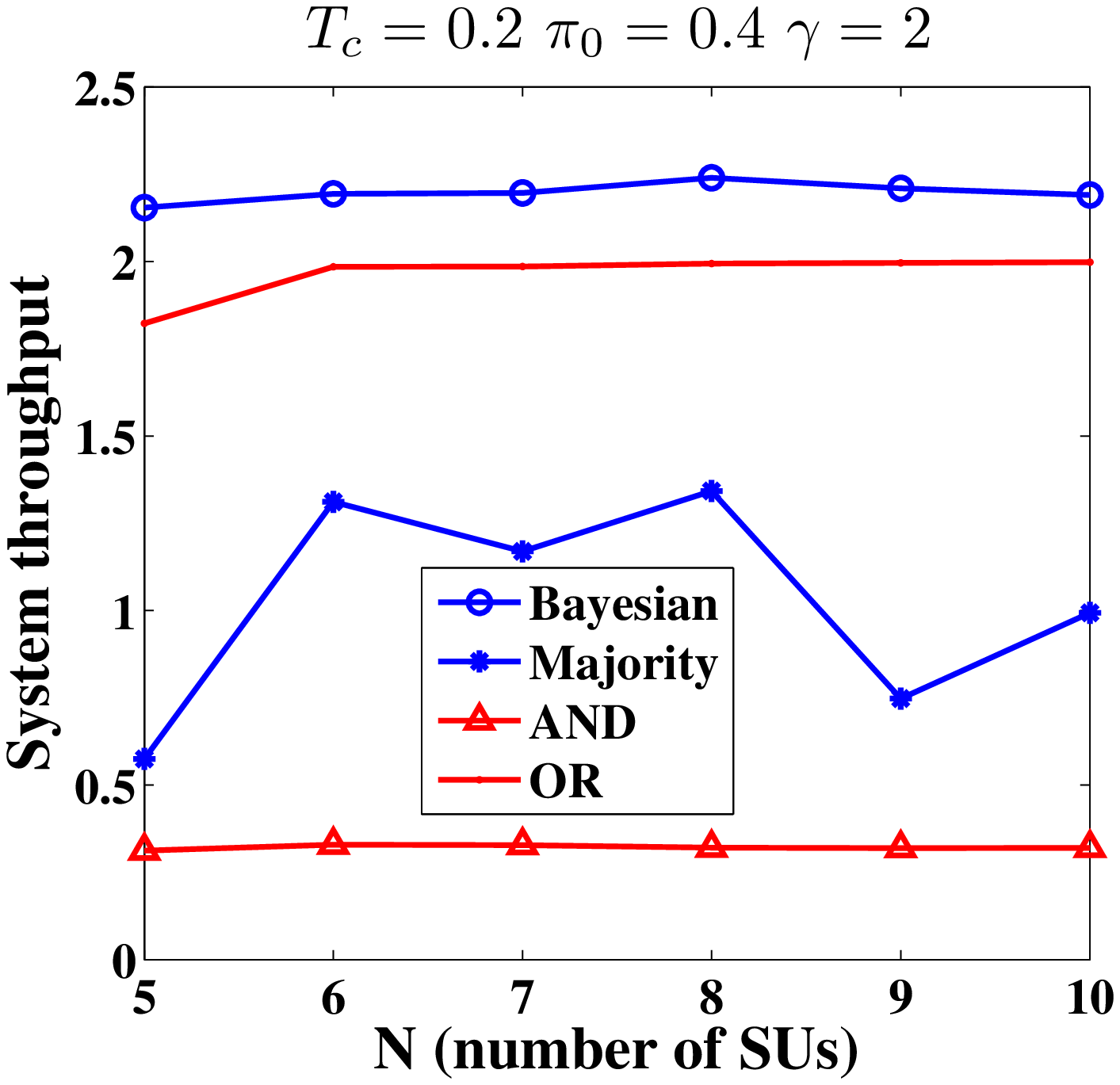}
\label{fig:overN_B}
}
\vspace{-0.5em}
\caption{Performance comparison of Bayesian decision rule, majority, AND and OR.}
\vspace{-1.5em}
\label{fig:all4}
\end{figure}

\subsection{Performance of Greedy Algorithm}
\label{subsec:B}

Greedy algorithm (Algorithm~\ref{alg:greedy}) can achieve strictly greater than $1/2$ of the optimal solution to Problem (B), as shown in Section~\ref{sec:PU_thru}. We compare its performance with that of random selection. Random selection is also based on Bayesian decision rule, which means Algorithm~\ref{alg:Bayesian} is first executed; after that, observations with $G(\boldsymbol{o})\ge H(\boldsymbol{o})$ are randomly selected to put in $O=1$ until the feasibility is satisfied. Thus the main difference between greedy algorithm and random selection lies in the selection criterion of observations with $G(\boldsymbol{o})\ge H(\boldsymbol{o})$ after the initial assignment based on Bayesian decision rule.

In addition, we set $\alpha=0.8$,  and $r=2$. We vary parameters such as $\gamma$, the average PU throughput, $\alpha$, the PU throughput constraint, $N$, the number of SUs, and $r$, the decimal places kept in Algorithm~\ref{alg:pseudo} in Figures~\ref{fig:gaN} and \ref{fig:r} respectively. To show the approximation factor of our algorithm accurately, two boundary cases are excluded in the result presentation where both greedy algorithm and random selection will give the optimal solution: 1) Bayesian decision rule gives the optimal solution; 2) It is optimal to put all observations in $O=1$. Hence, we only show their performance when at least one but not all observations with $G(\boldsymbol{o})\ge H(\boldsymbol{o})$ have to be moved to $O=1$. The average case and worst case performances are calculated based on the results after the exclusion.
 
In Figure~\ref{fig:gamma}, the approximation factors of greedy algorithm and random selection over the optimal solution are compared over different values of $\gamma$, the average PU throughput in the system. With a higher $\gamma$, the factor decreases gradually in both average and worst cases of Algorithm~\ref{alg:greedy} and it is the same with random selection although it fluctuates a bit due to the random selection. Greedy algorithm outperforms random selection in both average and worst cases. Potentially, the Bayesian decision rule assigns more SUs to $O=1$ compared to a lower $\gamma$ case. Thus the initial assignment is closer to $\alpha$, the PU throughput constraint. Since we only consider cases where Bayesian decision rule is not optimal, both algorithms tend to have worse performance when the initial assignment approaches $\alpha$ because it gets more sensitive to a wrong observation selection. 

\begin{figure*}[!t]
\centering
\subfigure[With different $\gamma$'s.]{
\includegraphics[scale=0.3]{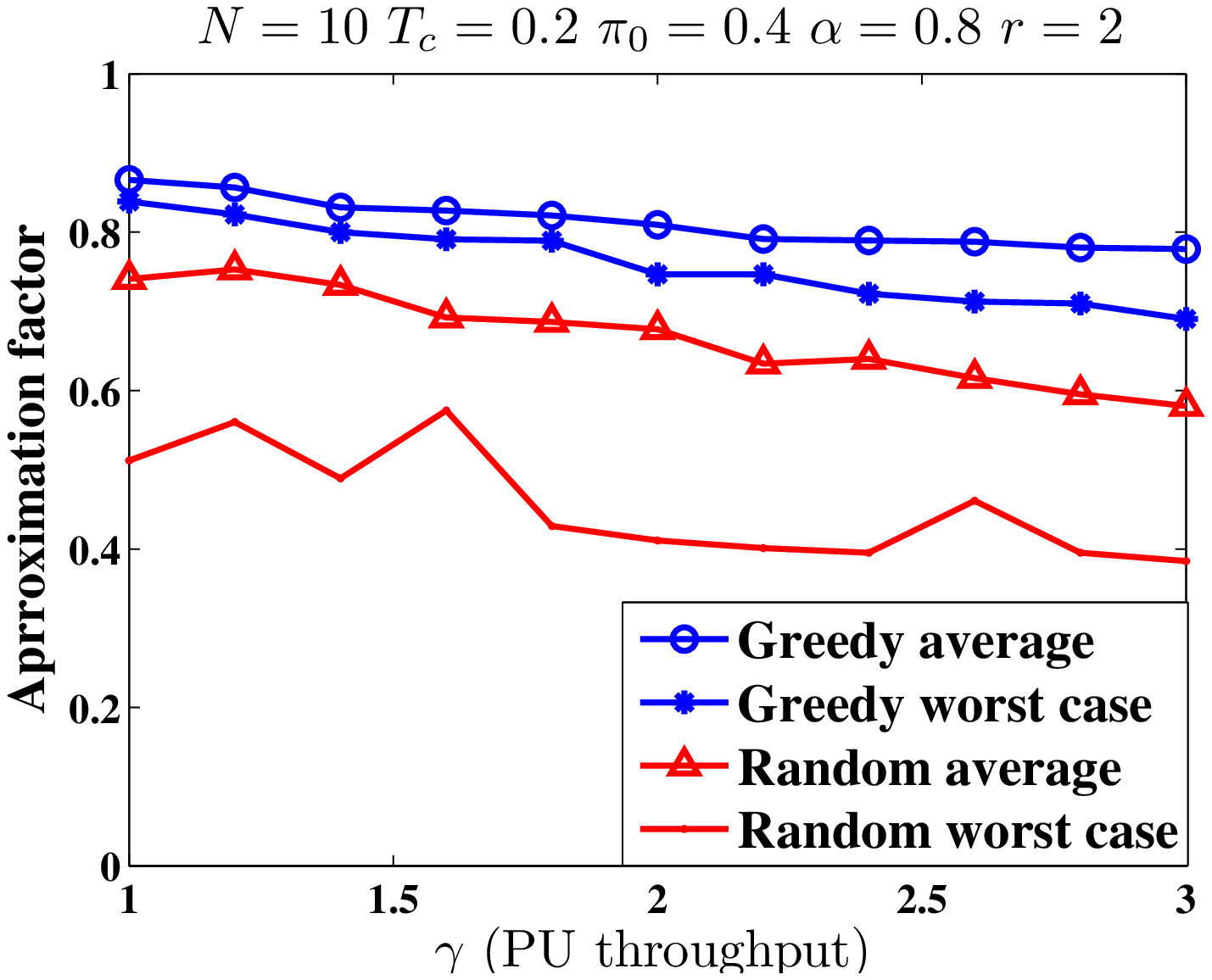}
\label{fig:gamma}
}
\hfil
\subfigure[With different $\alpha$'s.]{
\includegraphics[scale=0.32]{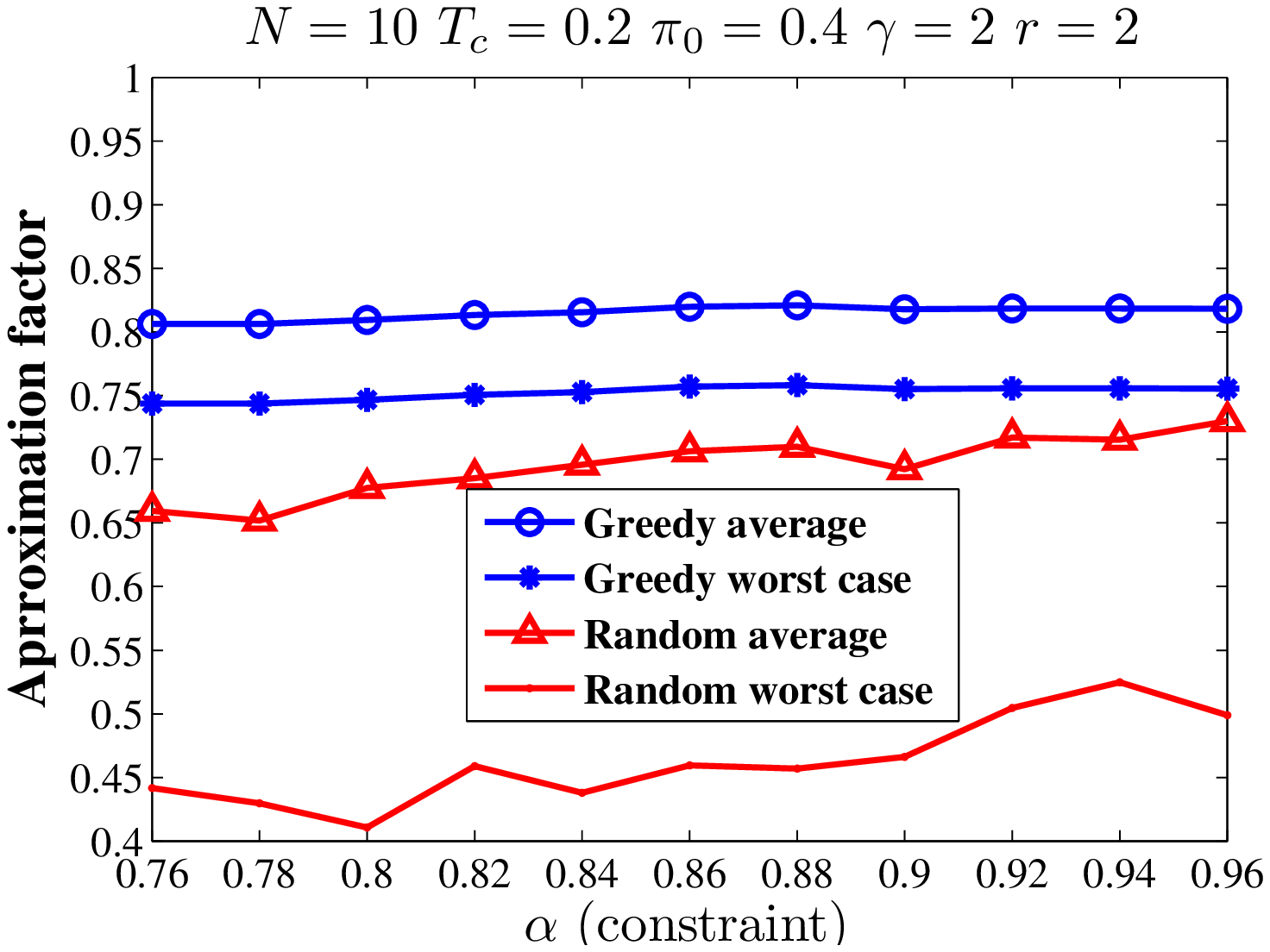}
\label{fig:alpha}
}
\hfil
\subfigure[With different $N$'s.]{
\includegraphics[scale=0.3]{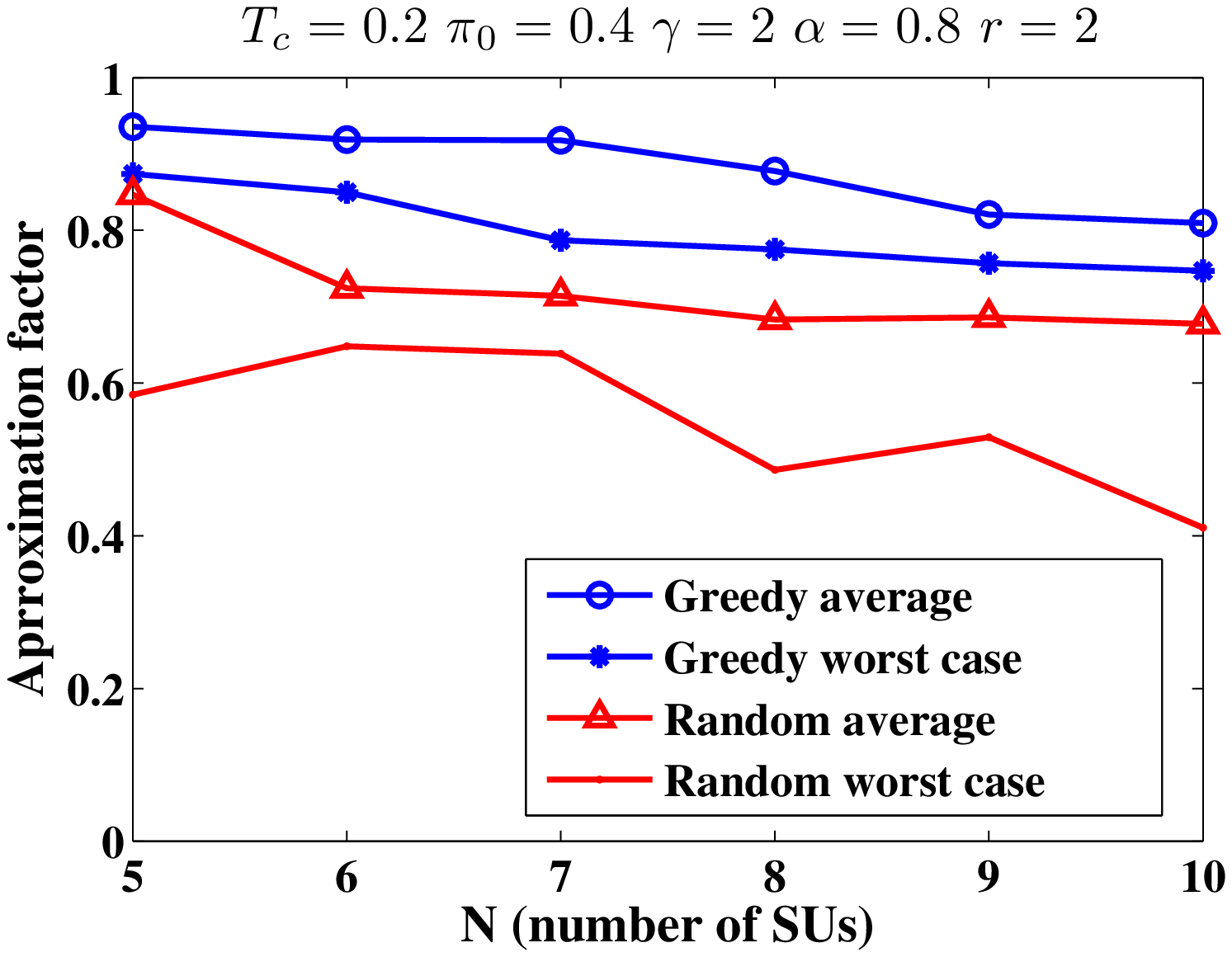}
\label{fig:overN}
}
\vspace{-1em}
\caption{Performance comparison of greedy algorithm and random selection.}
\vspace{-1.5em}
\label{fig:gaN}
\end{figure*}

In Figure~\ref{fig:alpha}, we vary $\alpha$, the PU throughput constraint, and compare the performance of greedy algorithm and random selection. Greedy algorithm is obviously better than random selection in both average and worst cases. Furthermore, the worst performance of all random runs generated in greedy algorithm wins over the average performance in random selection. Approximation factors in both of them increase, although it is minor in greedy algorithm. The increase can be explained similarly to that in Figure~\ref{fig:gamma}: a higher $\alpha$ makes the initial Bayesian decision assignment farther away from it so that the performance is less sensitive to the choice of observations. Due to its randomness, random selection may have poor performance with the factor as low as about $0.4$ in our simulation. 

We test the performance of greedy algorithm with different scales of the network. The number of SUs is varied from $5$ to $10$. The approximation factors of both algorithms degrade with more SUs. However, the factor is always far above $1/2$, as proved in Theorem~\ref{thm:greedy}. The random selection drops below $1/2$ in some cases as shown in the figure.   

In Algorithm~\ref{alg:pseudo}, we use $r$ as the decimal places kept for the calculations. Although the rounding error is not strictly characterized in Section III, we show the improvement of performance with higher $r$, which means higher resolution, in Figure~\ref{fig:r}. Note that the worst case performance is always greater than $1/2$. The average performance increases from about $0.7$ to about $0.95$, which is promising.   

\begin{figure}[tb]
    \begin{center}
    \setlength{\unitlength}{1in}
    \includegraphics[scale=0.3]{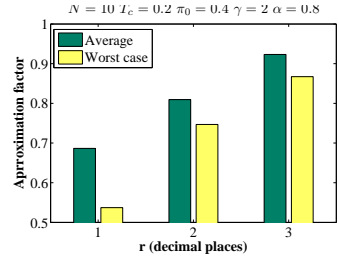}
    \end{center}
\vspace{-1.5em}
\caption{Performance comparison of greedy algorithm over different $r$'s.}
\vspace{-1.2em}
\label{fig:r}
\end{figure}

\subsection{Performance of Sequential Forward Selection}
\label{subsec:CD}

With $Nd+\tilde{d}\le T_c$, Problem (C) is the same as Problem (A) in that the optimal sensing set is the full set (Corollary~\ref{cor:full_best}). Thus, we focus on the performance of SFS, a heuristic for Problem (D) whose hardness is unknown so far. In SFS, we start from an empty sensing set. In every step, only the SU that is not yet chosen and has the largest marginal increase to the system throughput is added to the set. The algorithm stops when the size of the set reaches $k$. In Figure~\ref{fig:N_k}, we vary $k$, the size of the sensing set, from $1$ to $N$ and show the approximation factor of SFS compared to the optimal solution to Problem (D). When $k$ increases, the performance of SFS degrades until $k=N$ where all SUs are chosen in the sensing set so that the order of selection does not matter. SFS on average achieves at least $0.8$ of the optimal solution in our simulation although the factor is lower than $0.6$ in one of the worst cases.   

\begin{figure}[tb]
    \begin{center}
    \setlength{\unitlength}{1in}
    \includegraphics[scale=0.3]{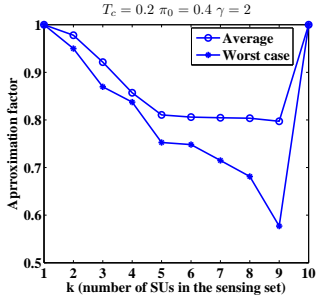}
    \end{center}
\vspace{-1.5em}
\caption{Performance comparison of SFS over different $k$'s.}
\vspace{-2em}
\label{fig:N_k}
\end{figure}

\section{conclusion}
\label{sec:con}
In this paper, we propose a series of algorithms to maximize the system throughput by cooperative sensing in cognitive radio networks. Bayesian decision rule is applied to solve the unconstrained optimization problem optimally. With the PU throughput constraint, the new problem is shown to be NP-hard and a greedy approximation algorithm with pseudo-polynomial time complexity is proposed. More importantly, the approximation factor is strictly greater than $1/2$. By restricting the number of SUs chosen for sensing, a new constrained optimization problem is formulated. We present a structural property that more SUs lead to better performance. However, the characterization of the combinatorial problem remains elusive, which is our future work. Moreover, we are also interested in investigating cases where the observations of SUs are correlated.  

\begin{spacing}{0.9}
\bibliographystyle{plain}
\small
\bibliography{Coop_Sens}

\begin{thebibliography}{10}

\bibitem{Berger}
J.~O. Berger.
\newblock {\em {Statistical Decision Theory and Bayesian Analysis}}.
\newblock Springer-Verlag, New York, 1985.

\bibitem{FCC}
Federal~Communications Commission.
\newblock Notice of proposed rulemaking, in the matter of unlicensed operation
  in the tv broadcast bands (docket no. 04-186) and additional spectrum for
  unlicensed devices below 900 mhzand in the 3 ghz band (02-380), fcc 04-113.
\newblock May 2004.

\bibitem{5054703}
T.~Do and B.L. Mark.
\newblock Joint spatial-temporal spectrum sensing for cognitive radio networks.
\newblock In {\em In proc. of CISS}, pages 124 --129, march 2009.

\bibitem{Fan_Jiang_2010}
R.~Fan and H.~Jiang.
\newblock Optimal multi-channel cooperative sensing in cognitive radio
  networks.
\newblock {\em IEEE Transactions on Wireless Communications}, 9(3):1128--1138,
  2010.

\bibitem{Ganesan05}
Y.G.~Li G.~Ganesan.
\newblock Cooperative spectrum sensing in cognitive radio networks.
\newblock In {\em In proc. of DySPAN}, pages 137 -- 143, Nov. 2005.

\bibitem{Garey}
Michael~R. Garey and David~S. Johnson.
\newblock {\em Computers and Intractability; A Guide to the Theory of
  NP-Completeness}.
\newblock W. H. Freeman \& Co., New York, NY, USA, 1990.

\bibitem{ieee}
http://grouper.ieee.org/groups/802/22/.
\newblock Ieee 802.22, working group on wireless regional area networks (wran).

\bibitem{survey}
R.~Balakrishnan I.~F.~Akyildiz, B. F.~Lo.
\newblock Cooperative spectrum sensing in cognitive radio networks: A survey.
\newblock {\em Physical Communication}, March 2011.

\bibitem{Kudo}
M.~Kudo.
\newblock {Comparison of algorithms that select features for pattern
  classifiers}.
\newblock {\em Pattern Recognition}, 33(1):25--41, jan 2000.

\bibitem{Lee_2008}
W.~Lee.
\newblock Optimal spectrum sensing framework for cognitive radio networks.
\newblock {\em IEEE Transactions on Wireless Communications}, 7(10):3845--3857,
  2008.

\bibitem{4927466}
B.~Mark and A.~Nasif.
\newblock Estimation of maximum interference-free power level for opportunistic
  spectrum access.
\newblock {\em Wireless Communications, IEEE Transactions on}, 8(5):2505
  --2513, may 2009.

\bibitem{Min}
A.~Min and K.~Shin.
\newblock An optimal sensing framework based on spatial rss-profile in
  cognitive radio networks.
\newblock In {\em In proc. of SECON}, pages 207--215, Piscataway, NJ, USA,
  2009. IEEE Press.

\bibitem{Mishra06}
S.~Mishra, A.t Sahai, and R.~W. Brodersen.
\newblock Cooperative sensing among cognitive radios.
\newblock In {\em In proc. of ICC}, pages 1658--1663, 2006.

\bibitem{5169958}
E.C.Y. Peh, Y.~Liang, Y.~Guan, and Y.~Zeng.
\newblock Optimization of cooperative sensing in cognitive radio networks: A
  sensing-throughput tradeoff view.
\newblock {\em Vehicular Technology, IEEE Transactions on}, 58(9):5294 --5299,
  nov. 2009.

\bibitem{Pena}
J.~M. Pena and R.~Nilsson.
\newblock On the complexity of discrete feature selection for optimal
  classification.
\newblock {\em IEEE Transactions on Pattern Analysis and Machine Intelligence},
  32:1517--1522, 2010.

\bibitem{Li2011}
S. Li, Z. Zheng, E. Ekici, and N. Shroff.
\newblock Technical report.
\newblock Maximizing System Throughput by Cooperative Sensing in Cognitive Radio Networks, 2011.
\newblock http://www.cse.ohio-state.edu/$\sim$lish/Shuang\_Li\_TR\_2011.pdf

\bibitem{4221472}
F.W. Seelig.
\newblock A description of the august 2006 xg demonstrations at fort a.p. hill.
\newblock In {\em In proc. of DySPAN}, pages 1 --12, april 2007.

\bibitem{Shahid_Kamruzzaman_2010}
M.~Shahid and J.~Kamruzzaman.
\newblock Weighted soft decision for cooperative sensing in cognitive radio
  networks.
\newblock {\em In proc. of ICON}, page~6, 2010.

\bibitem{SFS}
S.~D. Stearns.
\newblock {On selecting features for pattern classifiers.}
\newblock In {\em In proc. of ICPR}, pages 71--75, Coronado, CA, 1976.

\bibitem{4453896}
J.~Unnikrishnan and V.V. Veeravalli.
\newblock Cooperative sensing for primary detection in cognitive radio.
\newblock {\em Selected Topics in Signal Processing, IEEE Journal of}, 2(1):18
  --27, feb. 2008.

\bibitem{1447503}
H.~Urkowitz.
\newblock Energy detection of unknown deterministic signals.
\newblock {\em In proc. of the IEEE}, 55(4):523 -- 531, april 1967.

\bibitem{Van}
J.~M. Van~Campenhout.
\newblock {\em On the problem of measurement selection}.
\newblock PhD thesis, Stanford University, 1978.

\bibitem{Vijay}
V.~V. Vazirani.
\newblock {\em {Approximation Algorithms}}.
\newblock Springer, mar 2004.

\bibitem{1542650}
E.~Visotsky, S.~Kuffner, and R.~Peterson.
\newblock On collaborative detection of tv transmissions in support of dynamic
  spectrum sharing.
\newblock In {\em In proc. of DySPAN}, pages 338 --345, nov. 2005.

\bibitem{4533677}
Wei Zhang, R.K. Mallik, and K.~Ben~Letaief.
\newblock Cooperative spectrum sensing optimization in cognitive radio
  networks.
\newblock In {\em In proc. of ICC}, pages 3411 --3415, may 2008.

\end{thebibliography}
\end{spacing}


\end{document}